\documentclass[a4paper]{article}

\usepackage[utf8]{inputenc}
\usepackage{discopt}

\usepackage{mathtools}
\usepackage{todonotes} \tikzexternaldisable 
\usepackage{enumerate}
\usepackage{authblk}

\usepackage[hypertexnames=false,hidelinks]{hyperref}

\NewDocumentCommand\scalprod{mm}{\ensuremath{\left<#1,#2\right>}}
\NewDocumentCommand\oddSets{}{\ensuremath{\mathcal{U}}}

\newlength{\tagblockwidth}
\newlength{\tagblocksep}
\DeclareDocumentEnvironment{tagblock}{mmo}{
  \settowidth{\tagblockwidth}{(#1)}
  \IfValueTF{#3}{\setlength{\tagblocksep}{#3}}{\setlength{\tagblocksep}{0mm}}
  \par\noindent \parbox{\tagblockwidth}{\vspace{-1ex}\begin{align}\tag{#1}#2\end{align}}
  \hfill \begin{minipage}{\linewidth-\tagblockwidth-\tagblocksep-5mm}\vspace{-1ex}
}{\end{minipage} \vspace{2ex}}

\newtheorem{theorem}{Theorem}

\newtheorem{proposition}[theorem]{Proposition}

\newtheorem{corollary}[theorem]{Corollary}
\newtheorem{condition}[theorem]{Condition}

\title{A Note on Matchings Constructed during Edmonds' Weighted Perfect Matching Algorithm}
\author[1]{Volker Kaibel}
\affil[1]{Otto-von-Guericke-Universität Magdeburg, kaibel@ovgu.de}
\author[2]{Matthias Walter}
\affil[2]{RWTH Aachen University, walter@or.rwth-aachen.de}

\begin{document}

\maketitle
\thispagestyle{empty}

\begin{abstract}
  We reprove that all the matchings constructed during Edmonds' weighted perfect
  matching algorithm are optimal among those of the same cardinality 
  (provided that certain mild restrictions are obeyed on the choices the algorithm makes).
  We conclude that in order to solve a weighted matching problem it is not needed
  to solve a weighted perfect matching problem in an auxiliary graph of doubled size.
  This result was known before, e.g., posed as an exercise in see Lawler's book from 1976,
  but is not present in several modern books on combinatorial optimization.
\end{abstract}

\section{Introduction}

In this note we consider the minimum weight (perfect) matching problem.
An instance consists of an undirected simple graph $G = (V,E)$
with node set $V$ and edge set $E$ as well as weights $w \in \R^E$.
A \emph{matching} $M \subseteq E$ is a set of edges no two of which share a node.
A matching $M$ is \emph{perfect} if every node is covered, i.e., $|M| = \frac{1}{2} |V|$.
The \emph{matching number} $\nu(G)$ is the maximum cardinality of a matching in $G$.
The objective is to find a (perfect) matching $\widehat{M}$ of minimum weight 
$w(\widehat{M}) = \sum\limits_{e \in M} w_e$.
We call a matching \emph{cardinality-optimal} if it has minimum weight among all matchings of
the same cardinality.

\medskip

The set of edges incident to a node $v \in V$ will be denoted by $\delta(v)$
and for some node set $U \subseteq V$ the set of edges inside $U$ is $E[U]$.
The set $\oddSets = \setdef{ U \subseteq V }[ |U| \text{ odd} ]$ 
consists of all node subsets of odd cardinality.
The characteristic vector $\chi(M) \in \setdef{0,1}^E$ of a subset $M \subseteq E$ is the vector
with $\chi(M)_e = 1$ if and only if $e \in M$.
For a vector $x \in \R^E$ and some set $F \subseteq E$ we write $x(F)$ for $\sum_{e \in F} x_e$.
We assume that the reader is familiar with the matching algorithm, basic matching theory, 
and with linear programming concepts.
For further material, see \cite{Schrijver86}.

\medskip

In 1965, Edmonds devised a combinatorial algorithm for cardinality-maximum matching \cite{Edmonds65a}.
In the same year he carried out the analysis of the matching polyhedron yielding a version of the algorithm that
can solve weighted matching problems as well (see \cite{Edmonds65b}).
Lawler poses an exercise in his book from 1976 (see Problem~8.2~in~\cite{Lawler76}): Prove that the intermediate matchings
constructed in the algorithm are cardinality-optimal.
His version of Edmonds' algorithm is based on the so-called blossom inequalities
which makes this problem easy to solve.
In contrast to this, more modern standard books on combinatorial optimization (e.g., \cite{Schrijver03}, \cite{KorteV12})
are based on the cut inequalities (which are only valid for the perfect-matching polytope),
and for which the technical effort for a proof based on linear programming is higher.

We use the the version of Edmonds' algorithm presented in Chapter~26 of Schrijver's book \cite{Schrijver03}
and only briefly recall the algorithm and important properties in Section~\ref{SectionAlgorithmProperties}.
We then discuss the matching polytope, adjacency properties, and linear programming formulations.
In Section~\ref{SectionIntermediateMatchings} we present the proof of the mentioned result (Theorem~\ref{TheoremIntermediateCardOpt}) stating
that the intermediate matchings
constructed in the algorithm are cardinality-optimal.
We provide two proofs, one by constructing a dual-feasible solution, 
and another one by extending the current graph.
In both cases we prove optimality using complementary slackness.
We close by discussing the imposed conditions as well as applications of our main result.

Before going into the details of the algorithm,
observe that we can assume the edge weights $w$ to be nonnegative,
since otherwise we can add a large constant $C$ to every weight.
Then the weight of every matching of cardinality $k$ is increased by the same value $k \cdot C$,
leaving the ordering of these matchings intact.

\pagebreak[3]
\section{Edmonds' Algorithm}
\label{SectionAlgorithmProperties}

In order to setup the notation we present the algorithm according 
to Schrijver's book \cite{Schrijver03} skipping the analysis.

The input is an undirected graph $G = (V,E)$ with nonnegative edge weights $w \in \R_+^E$.
The algorithm is primal-dual where the primal state is given by a current matching $M$ in $G$
(with primal variable $x = \chi(M)$).
The dual state is given by a laminar collection $\Omega \subseteq \oddSets$
(i.e., for each $U,U' \in \Omega$ we either have $U \cap U' = \emptyset$ or $U \subseteq U'$ or $U' \subseteq U$) 
with associated dual variables $\pi : \oddSets \to \R$ having support only in $\Omega$, i.e., $\pi(U) = 0$ for all $U \notin \Omega$.
We can state the constraints for $\pi$ as follows:
\begin{align}
  \pi(U)              &\geq 0     && \forall U \text{ with } |U| \geq 3 \label{EquationMatchingDualCutNonnegative} \\
  \sum_{U: e \in \delta(U)} \pi(U) &\leq w_e && \forall e = \setdef{u,v} \in E \label{EquationMatchingDualCutEdge}
\end{align}
All singleton sets $\setdef{v}$ will be contained in the collection $\Omega$ throughout the algorithm
and hence the inclusion-wise maximal sets in $\Omega$, denoted by $\Omega_{\max}$, form a partition of $V$.
The main auxiliary graph is $G'$ which is obtained by 
only considering edges $\setdef{u,v} \in E$ which satisfy \eqref{EquationMatchingDualCutEdge} with equality
and shrinking all node sets $U \in \Omega_{\max}$ to single nodes.

\smallskip

The algorithm initializes $M := \emptyset$ and $\pi := \zerovec[\oddSets]$.
Throughout the algorithm $\pi$ will always be feasible and complementary slackness will always
be preserved. As soon as the matching $M$ becomes feasible (i.e., a perfect matching)
it is also optimal.

\smallskip

The algorithm attempts to do primal updates or, if these are not possible, dual updates
and repeats until a perfect matching is found or it asserts that no perfect matching exists.

For the primal update, the set $X \subset \Omega_{\max}$ of unmatched nodes in $G'$ is considered.
We try to find an $M$-alternating walk in $G'$ from a node in $X$ to another node in $X$.
In case this walk is a path, $M$ gets augmented, otherwise a blossom is found and its associated odd set $U$ is then shrunken
and added to $\Omega$.

In case no such walk is found, the disjoint node sets $\mathcal{S}, \mathcal{T} \subseteq \Omega_{\max}$ of $G'$
are computed which contain the nodes to which $G'$ has an odd (resp. even) length
$M$-alternating path starting from any node in $X$.
Then the $\pi$-values of those nodes $U$ in $\mathcal{S}$ (resp. $\mathcal{T}$) are decreased (resp. increased)
by the largest $\alpha$ such that the constraints 
\eqref{EquationMatchingDualCutNonnegative} and \eqref{EquationMatchingDualCutEdge} are still satisfied.
If $\alpha$ can be chosen arbitrarily large, then the algorithm asserts that $G$ has no perfect matching.
At the end of the phase all $U \in \Omega_{\max}$ with $|U| \geq 3$ for which $\pi(U)$ became
$0$ are deshrunken and are removed from $\Omega$.

\bigskip

The next proposition states that the case of arbitrarily large $\alpha$
occurs as late as possible.

\begin{proposition}
  \label{TheoremTermination}
  When the algorithm asserts that $G$ has no perfect matching, then the current matching has maximum cardinality,
  that is, $|\widetilde{M}| = \nu(G)$.
\end{proposition}

\begin{proof}
  Suppose we are in the situation that $\alpha$ can be chosen arbitrarily large.
  
  First observe that the nodes in $\mathcal{S}$ must be singletons, i.e., are not shrunken,
  since their $\pi$-values are decreased by $\alpha$ and for non-singletons Inequality~\eqref{EquationMatchingDualCutNonnegative}
  would restrict the decrease.
  Furthermore, they are matched to nodes in $\mathcal{T}$ by definition of $\mathcal{S}$ and $\mathcal{T}$.

  Second, nodes in $\mathcal{T}$ only have neighbors (in $G'$) in $\mathcal{S}$
  since their $\pi$-values are increased by $\alpha$, but Inequality~\eqref{EquationMatchingDualCutEdge} for the edge in question
  does not restrict $\alpha$.

  Now assume, for the sake of contradiction, that $|\widetilde{M}| < \nu(G)$ holds and hence,
  there exists a $\widetilde{M}$-augmenting path $P$ in $G$ connecting some $\widetilde{M}$-exposed nodes $s,t \in X$.

  From this it follows that every $\widetilde{M}$-edge of $P$ is either an edge in $G'$ connecting a node from $\mathcal{S}$
  with a node from $\mathcal{T}$, or is an edge inside a blossom (whose representative node is in $\mathcal{T}$).
  But this already contradicts the fact that $P$ is $\widetilde{M}$-augmenting since there is no edge between
  two nodes from $\mathcal{T}$.
\end{proof}

This version of the algorithm clearly satisfies the following condition which is important as we will see later.

\begin{condition}
  \label{ConditionSameInit}
  The dual values are initialized by $\pi(U) := 0$ for $U$ with $|U| \geq 3$
  and $\pi(\setdef{v}) := \beta$ for all $U = \setdef{v}$ with $v \in V$ for some fixed $\beta \in \R$.

  Furthermore, in a dual update step, the dual values $\pi(U)$ are increased
  by the same amount for all $U \in \mathcal{T}$ for which the values are increased.
\end{condition}

\medskip

Let for every node $v \in V$ the accumulated dual value be denoted by $\pi^*(v) := \sum_{U \in \oddSets, v \in U} \pi(U)$
and let $\pi^*_{\max} := \max_{v \in V} \pi^*(v)$ be their maximum.
Let at any stage $\widetilde{M}$ be the matching that one obtains from the current matching $M$
by iteratively deshrinking the sets in $\Omega_{\max}$.

\begin{proposition}
  \label{TheoremProperties}
  At any stage of the algorithm, the matching $\widetilde{M}$ and the dual values $\pi$
  satisfy these properties:
  \begin{enumerate}[(i)]
  \item\label{TheoremPropertiesBlossom}
    If $\pi(U) > 0$ for some set $U \in \oddSets{}$ with $|U| \geq 3$,
    then $U \in \Omega$ satisfies 
    \begin{align}
      \left|\widetilde{M} \cap E[U]\right| = \frac{1}{2}(|U|-1) \label{EquationNearPerfect} \ .
    \end{align}
  \item\label{TheoremPropertiesAugmentation}
    Shrinking and deshrinking do not modify the current matching $\widetilde{M}$.
    Furthermore, augmenting $M$ in $G'$ corresponds to an augmentation of $\widetilde{M}$
    along a single path.
  \item\label{TheoremPropertiesMatchedRemain}
    Matched nodes will always remain matched.
  \item\label{TheoremPropertiesUnmatchedIncrease}
    The dual values $\pi(U)$ for \emph{unmatched} nodes $U$ of $G'$
    are always increased in a dual update.
  \end{enumerate}
\end{proposition}

The next corollary follows readily from 
Proposition~\ref{TheoremProperties}~\eqref{TheoremPropertiesMatchedRemain}~and~\eqref{TheoremPropertiesUnmatchedIncrease}.

\begin{corollary}
  \label{TheoremUnmatchedMaximumPi}
  Given Condition~\ref{ConditionSameInit},
  at any stage $\pi^*(v) = \pi^*_{\max}$ holds for every unmatched node $v \in X$.
\end{corollary}

\pagebreak[3]
\section{The Matching Polytope}
\label{SectionMatchingPolytope}

Let $P(G) = \conv\setdef{\chi(M)}[ M \text{ matching in } G]$ be the \emph{matching polytope} of a graph $G$.
Here $\chi(M) \in \setdef{0,1}^E$ denotes the characteristic vector which has $\chi(M)_e = 1$ for $e \in E$ 
if and only if $e \in M$.
Edmonds \cite{Edmonds65b} gave the following outer description of the polytope $P(G)$:
\begin{align}
  x_e           &\geq 0                   & \forall e \in E \label{EquationMatchingPrimalNonnegative} \\
  x(\delta(v))  &\leq 1                   & \forall v \in V \label{EquationMatchingPrimalStar} \\
  x(E[U])       &\leq \frac{|U|-1}{2}     & \forall U \in \oddSets \label{EquationMatchingPrimalSetBlossom}
\end{align}
Chv{\'a}tal \cite{Chvatal75} characterized adjacency in $P(G)$:
\begin{proposition}
  \label{TheoremAdjacency}
  The vertices corresponding to two matchings $M$ and $M'$ are adjacent in $P(G)$ if and only if
  $M \Delta M'$ is connected.
\end{proposition}

\begin{corollary}
  The sequence of matchings constructed during the algorithm corresponds to the vertex sequence
  of a path in the matching polytope. The vertices correspond to matchings of strictly increasing cardinality.
\end{corollary}

\begin{proof}
  This follows directly from 
  Proposition~\ref{TheoremProperties}~\eqref{TheoremPropertiesBlossom}~and~\eqref{TheoremPropertiesAugmentation}.
\end{proof}

Another direct consequence is that  $|M| - |M'| \in \setdef{-1,0,+1}$ holds 
for adjacent vertices $\chi(M)$, $\chi(M')$.
Hence, for every $k \in \Z_+$ we have that the convex hull $P_k$ of all matchings of cardinality $k$
is equal to $P$ intersected with all $x \in \R^E$ satisfying
\begin{align}
  x(E) &= k \label{EquationMatchingPrimalCardinality} \ ,
\end{align}
since such a hyperplane does not intersect the relative interior of any edge of $P(G)$.
We now state the corresponding linear program

\begin{tagblock}{P}{\label{MatchingPrimalCardinality}}
  \begin{align*}
    \min
    \scalprod{w}{x} 
    \text{ s.t. \eqref{EquationMatchingPrimalNonnegative}, \eqref{EquationMatchingPrimalStar},
    \eqref{EquationMatchingPrimalSetBlossom}, and \eqref{EquationMatchingPrimalCardinality} }
  \end{align*}
\end{tagblock}
and its dual:
\begin{tagblock}{D}{\label{MatchingDualCardinality}}
  \begin{flalign}
    \min
    & ~\mathrlap{\sum\limits_{v \in V} y_v 
        + \sum\limits_{U \in \,\oddSets{}} \frac{|U|-1}{2} z_U
        + k \gamma} \nonumber \\
    \text{s.t.} 
    && y_v + y_w + \sum\limits_{\substack{U \in \,\oddSets\\e \in E[U]}} z_U + \gamma
            & \leq w_e  && \forall e = \setdef{v,w} \in E && \label{EquationMatchingDualSetEdge} \\
    && y_v  & \leq 0    && \forall v \in V && \label{EquationMatchingDualSetNodeNonpositive} \\
    && z_U  & \leq 0    && \forall U \in \oddSets \label{EquationMatchingDualSetBlossomNonpositive}
  \end{flalign}
\end{tagblock}
Here $y$, $z$, $\gamma$ correspond to Inequalities 
\eqref{EquationMatchingPrimalStar}, \eqref{EquationMatchingPrimalSetBlossom}, and Equation~\eqref{EquationMatchingPrimalCardinality}, respectively.

For the case of perfect matchings, i.e., $k = \frac{1}{2}|V|$, there is an equivalent formulation of $P_k$
using blossom inequalities in \emph{cut form} requiring~\eqref{EquationMatchingPrimalStar},~\eqref{EquationMatchingPrimalCardinality} and 
\begin{align}
  x(\delta(U)) & \geq 1 ~~~~\forall U \in \oddSets \label{EquationMatchingPrimalCutBlossom}
\end{align}
instead of \eqref{EquationMatchingPrimalSetBlossom}.

\pagebreak[3]
\section{Intermediate Matchings}
\label{SectionIntermediateMatchings}

We consider the state of the algorithm in any stage.
Let $k := |\widetilde{M}|$ and consider the two linear programs for finding cardinality-optimal matchings.
Clearly, $x := \chi(\widetilde{M})$ is a feasible solution for \eqref{MatchingPrimalCardinality} since it is a matching of
the right cardinality.
The setup constructed so far prepares us for proving our main result.

\begin{theorem}
  \label{TheoremIntermediateCardOpt}
  Given Condition~\ref{ConditionSameInit},
  every matching $\widetilde{M}$ constructed during Edmonds' weighted matching algorithm
  is cardinality-optimal.
\end{theorem}

Our proof strategy is to construct a dual feasible solution $y,z,\gamma$
and then proof complementary slackness.
Clearly, dual values for Inequalities~\eqref{EquationMatchingPrimalSetBlossom}
can be calculated from those of Inequalities~\eqref{EquationMatchingPrimalCutBlossom}
by following the transformation in the proof showing their equivalence. 
This leads to the following formulas:
  
\begin{equation}
  \label{DualTransformation} 
  \begin{array}{rcll}
    \gamma & := & 2 \pi^*_{\max} \\
    y_v & := & \pi^*(v) - \pi^*_{\max} & \forall v \in V \\
    z_U & := & \multicolumn{2}{c}{ \begin{cases} 
      -2 \pi(U)~~~~~ & \forall U \in \Omega \textrm{ with } |U| \geq 3 \\
      0 & \forall U \in \oddSets \textrm{ with } |U| = 1 \textrm{ or } U \notin \Omega
    \end{cases}}
  \end{array}
\end{equation}
Here, again, $\pi^*(v) := \sum_{U \in \Omega, v \in U} \pi(U)$ and $\pi^*_{\max} := \max_{v \in V} \pi^*(v)$.
Using this transformation we can easily proof Theorem~\ref{TheoremIntermediateCardOpt}.

\begin{proof}[Proof of Theorem~\ref{TheoremIntermediateCardOpt}]
Let $\pi(U)$ for $U \in \Omega$ be the dual values at the current stage of the algorithm,
and let $k := |\widetilde{M}|$ be the cardinality of the current matching.
We construct a solution to the dual LP
by applying the transformation \eqref{DualTransformation}.
  
We first prove that $(y,z,\gamma)$ is feasible for \eqref{MatchingDualCardinality}.
Clearly $y \leq \zerovec[V]$ and $z \leq \zerovec[\oddSets]$.
Furthermore, observe for every edge $e = \setdef{u,v}$ the relation
\begin{align}
  \label{EquationNewDualInequality}
  y_u + y_v + \hspace{-2mm} \sum\limits_{\substack{U \in \,\oddSets \\ e \in E[U]}} \hspace{-2mm} z_U + \gamma
  = \sum\limits_{\substack{U \in \,\oddSets \\ u \in U}} \pi(U)
  + \sum\limits_{\substack{U \in \,\oddSets \\ v \in U}} \pi(U)
  - 2 \hspace{-2mm} \sum\limits_{\substack{U \in \,\oddSets \\ e \in E[U]}} \hspace{-2mm} \pi(U)
  = \hspace{-2mm} \sum\limits_{\substack{U \in \,\oddSets \\ e \in \delta(U)}} \hspace{-2mm} \pi(U)
  &\leq w_e
\end{align}
is satisfied, proving feasibility.
It also proves that Inequality~\eqref{EquationMatchingDualSetEdge} is satisfied with equality
for edges $e$ with $x_e > 0$ since in this case the algorithm ensures that also
Inequality~\eqref{EquationMatchingDualCutEdge} is satisfied with equality.

Second, assume $z_U < 0$, that is, $\pi(U) > 0$ for some $U \in \oddSets$ with $|U| \geq 3$.
Then, by Proposition~\ref{TheoremProperties}~\eqref{TheoremPropertiesBlossom},
the corresponding Inequality \eqref{EquationMatchingPrimalSetBlossom} is tight.
Third, if $x(\delta(v)) < 1$ holds, i.e., $v$ is not matched by $\widetilde{M}$,
then by Corollary~\ref{TheoremUnmatchedMaximumPi} we have $y_v = \pi^*(v) - \pi^*_{\max} = 0$.
This concludes the proof.
\end{proof}

\subsection*{Alternative Proof}

Jens Vygen suggested the following alternative proof idea (personal communication).
Let $v_1, \ldots, v_k$ be all unmatched nodes at the current stage.
We construct an auxiliary graph $G' = (V', E')$ with $V' = V ~\cupdot~ \setdef{ u_1, \ldots, u_k }$
and $E' = E ~\cupdot~ \setdef{ \setdef{v,u_i} }[ v \in V, i \in [k] ]$ (see Figure~\ref{FigureAlternativeProof}).
We define weights $w'_e = w_e$ for $e \in E$ and $w'_e = 0$ for $e \in E' \setminus E$.
Then every matching $M$ in $G$ can be extended to a perfect matching $M'$ in $G'$
of the same weight.

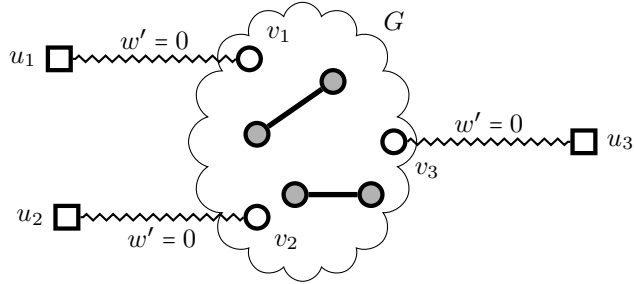
\begin{figure}[ht!]
  \begin{center}
  \begin{tikzpicture}
    \tikzset{
      exposed node/.style={circle,draw=black,ultra thick,fill=white,inner sep=0mm,minimum size=3.0mm},
      matched node/.style={circle,draw=black,ultra thick,fill=black!30!white,inner sep=0mm,minimum size=3.0mm},
      auxiliary node/.style={rectangle,draw=black,ultra thick,fill=white,inner sep=0mm,minimum size=3.0mm},
      matching edge/.style={line width=2pt},
      auxiliary edge/.style={thick,decorate,decoration={zigzag,amplitude=1.2pt,segment length=5pt}},
    }
    
    \node[matched node] (m1a) at (0,0.3) {};
    \node[matched node] (m1b) at (1,1) {};
    \node[matched node] (m2a) at (0.5,-0.5) {};
    \node[matched node] (m2b) at (1.5,-0.5) {};
    \draw[matching edge] (m1a) -- (m1b);
    \draw[matching edge] (m2a) -- (m2b);

    \node[exposed node] (v1) at (-0.1,+1.3) {};
    \node[exposed node] (v2) at (0,-0.8) {};
    \node[exposed node] (v3) at (1.8,+0.2) {};

    \node[cloud, cloud puffs=16, draw, minimum width=30mm, minimum height=37mm] at (0.6,0.2) {};
    \node at (1.8,1.8) {$G$};

    \node[auxiliary node] (u1) at (-2.6,+1.3) {};
    \node[auxiliary node] (u2) at (-2.5,-0.8) {};
    \node[auxiliary node] (u3) at (+4.3,+0.2) {};
    \draw[auxiliary edge] (v1) -- node[auto,swap] {$w' = 0$} (u1);
    \draw[auxiliary edge] (v2) -- node[auto] {$w' = 0$} (u2);
    \draw[auxiliary edge] (v3) -- node[auto] {$w' = 0$} (u3);

    \node[above right=-1pt of v1] {$v_1$};
    \node[below right=-1pt of v2] {$v_2$};
    \node[below right=1pt of v3] {$v_3$};
    
    \node[left=0pt of u1] {$u_1$};
    \node[left=0pt of u2] {$u_2$};
    \node[right=0pt of u3] {$u_3$};
  \end{tikzpicture}
  \end{center}
  \caption{Alternative proof of Theorem~\ref{TheoremIntermediateCardOpt} via auxiliary graph.}
  \label{FigureAlternativeProof}
\end{figure}

The extended version $\widetilde{M}'$ of the current matching $\widetilde{M}$
(from the algorithm) is a minimum perfect matching using the following dual values:
\begin{equation*}
  \begin{array}{rcll}
    \pi'(U) & := & \pi(U) & \forall U \in \Omega \\
    \pi'(\setdef{u_i}) &:=& -\pi^*_{\max} &\forall i \in [k]
  \end{array}
\end{equation*}
Then $\pi'$ is dual feasible since $\pi$ was dual feasible in $G$
and since $\pi^*(v_i) = \pi^*_{\max}$ for all $i \in [k]$ we have
$\pi^*(v_i) + \pi^*(u_i) = 0 = w'_e$ for all edges $e = \setdef{v_i, u_i}$, $i \in [k]$.
Furthermore, these edges are tight and hence $\widetilde{M}'$ consists of only tight edges.
Clearly all $U$ have some leaving edges since $\widetilde{M}'$ is perfect, that is,
the complementary slackness conditions are satisfied, too.

\pagebreak[3]
\section{Discussion}
\label{SectionDiscussion}

From our main theorem it readily follows that one can run the algorithm on any graph
and obtain a cardinality-optimal matching of every possible cardinality.
From those we can then select the optimum matching and hence solve the minimum weight matching problem
using the minimum weight perfect matching algorithm.
This is different from the usual construction auf an auxiliary graph $\widetilde{G}$ 
of twice the size in order to solve a matching problem on $G$
by solving a perfect matching problem on $\widetilde{G}$.

\medskip

Unfortunately, Condition~\ref{ConditionSameInit} seems to be unattractive for practical usage.
Cook and Rohe \cite{CookR98} invented the idea 
not to modify the dual values by the same amount on all nodes in question
in order to allow greater changes in different parts of the graph.
Since then, for state of the art implementations (e.g. \emph{Blossom V}, \cite{Kolmogorov09})
our imposed condition does not hold. 

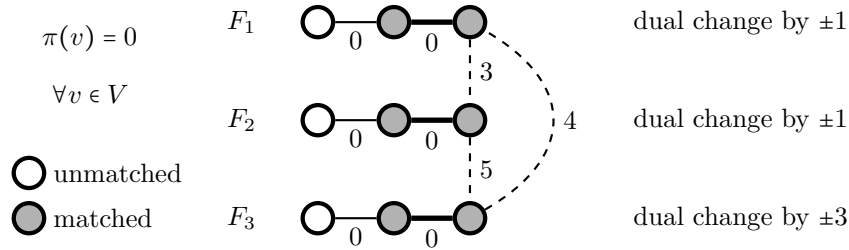
\begin{figure}[ht!]
  \begin{center}
  \begin{tikzpicture}
    \tikzset{
      exposed node/.style={circle,draw=black,ultra thick,fill=white,inner sep=0mm,minimum size=4.0mm},
      matched node/.style={circle,draw=black,ultra thick,fill=black!30!white,inner sep=0mm,minimum size=4.0mm},
      present edge/.style={thick},
      matching edge/.style={line width=2pt},
      missing edge/.style={thick,dashed},
    }

    \node[matched node] (f1a) at (0,+1.3) {};
    \node[matched node] (f1b) at (-1,+1.3) {};
    \node[exposed node] (f1c) at (-2,+1.3) {};
    \node[left of=f1c] {$F_1$};
    \node[right=18mm of f1a] {dual change by $\pm1$};
    
    \node[matched node] (f2a) at (0,0) {};
    \node[matched node] (f2b) at (-1,0) {};
    \node[exposed node] (f2c) at (-2,0) {};
    \node[left of=f2c] {$F_2$};
    \node[right=18mm of f2a] {dual change by $\pm1$};
    
    \node[matched node] (f3a) at (0,-1.3) {};
    \node[matched node] (f3b) at (-1,-1.3) {};
    \node[exposed node] (f3c) at (-2,-1.3) {};
    \node[left of=f3c] {$F_3$};
    \node[right=18mm of f3a] {dual change by $\pm3$};

    \draw[matching edge] (f1a) -- node[auto] {$0$} (f1b);
    \draw[matching edge] (f2a) -- node[auto] {$0$} (f2b);
    \draw[matching edge] (f3a) -- node[auto] {$0$} (f3b);
    \draw[present edge]  (f1b) -- node[auto] {$0$} (f1c);
    \draw[present edge]  (f2b) -- node[auto] {$0$} (f2c);
    \draw[present edge]  (f3b) -- node[auto] {$0$} (f3c);
    \draw[missing edge]  (f1a) to node[auto] {$3$} (f2a)
                               to node[auto] {$5$} (f3a)
                               to[bend angle=60,looseness=1.5,bend right] node[auto,swap] {$4$} (f1a);

    \node at (-5,1) {\begin{minipage}{20mm}$$\pi(v) = 0$$ $$\forall v \in V$$\end{minipage}};
    
    \node[anchor=west] at (-5.6,-1.3) {matched};
    \node[matched node] at (-5.8,-1.3) {}; 
    
    \node[anchor=west] at (-5.6,-0.7) {unmatched};
    \node[exposed node] at (-5.8,-0.7)  {}; 
  \end{tikzpicture}
  \end{center}
  \caption{An example showing necessity of Condition~\ref{ConditionSameInit}.}
  \label{FigureConditionNecessary}
\end{figure}

In the following example (see Figure~\ref{FigureConditionNecessary})
we have three alternating forests of which all edges have weight $0$.
All dual variables are $0$ as well.
According to the above mentioned ``variable dual changes'' approach we can,
for whatever reason, modify the values of the forests $F_1$, $F_2$, $F_3$ by
$1$, $1$, and $3$, respectively.
In the next step only the edge of weight $4$ has become tight, leading to a suboptimal augmentation.

\bigskip

\noindent
\textbf{Ackknowledgements.} We are thankful to Jens Vygen for comments
on an earlier version of the manuscript, among those the alternative proof
of Theorem~\ref{TheoremIntermediateCardOpt} given above.

\bibliographystyle{abbrv}
\bibliography{references}

\begin{thebibliography}{1}

\bibitem{Chvatal75}
V.~Chv{\'a}tal.
\newblock On certain polytopes associated with graphs.
\newblock {\em Journal of Combinatorial Theory, Series B}, 18(2):138 -- 154,
  1975.

\bibitem{CookR98}
W.~Cook and A.~Rohe.
\newblock Computing minimum-weight perfect matchings.
\newblock {\em INFORMS J. on Comp.}, 11:138--148, 1998.

\bibitem{Edmonds65b}
J.~Edmonds.
\newblock Maximum matching and a polyhedron with 0, 1-vertices.
\newblock {\em Journal of Research of the National Bureau of Standards B},
  69:125--130, 1965.

\bibitem{Edmonds65a}
J.~Edmonds.
\newblock Paths, trees, and flowers.
\newblock {\em Canad. J. Math.}, 17:449--467, 1965.

\bibitem{Kolmogorov09}
V.~Kolmogorov.
\newblock {Blossom V: a new implementation of a minimum cost perfect matching
  algorithm}.
\newblock {\em Math. Program. Ser. C}, 1:43--67, 2009.

\bibitem{KorteV12}
B.~Korte and J.~Vygen.
\newblock {\em Combinatorial Optimization}, volume~21.
\newblock Springer, fifth edition, 2012.

\bibitem{Lawler76}
E.~Lawler.
\newblock {\em Combinatorial Optimization: Networks and Matroids}.
\newblock Saunders College Publishing, Fort Worth, 1976.

\bibitem{Schrijver86}
A.~Schrijver.
\newblock {\em Theory of linear and integer programming}.
\newblock John Wiley \& Sons, Inc., New York, NY, USA, 1986.

\bibitem{Schrijver03}
A.~Schrijver.
\newblock {\em Combinatorial Optimization - Polyhedra and Efficiency}.
\newblock Springer, 2003.

\end{thebibliography}

\end{document}